\documentclass[letterpaper]{article} % DO NOT CHANGE THIS
\usepackage{aaai25}  % DO NOT CHANGE THIS
\usepackage{times}  % DO NOT CHANGE THIS
\usepackage{helvet}  % DO NOT CHANGE THIS
\usepackage{courier}  % DO NOT CHANGE THIS
\usepackage[hyphens]{url}  % DO NOT CHANGE THIS
\usepackage{amsthm,amsmath,amssymb, cleveref}

\usepackage{graphicx} % DO NOT CHANGE THIS
\usepackage{natbib}  % DO NOT CHANGE THIS AND DO NOT ADD ANY OPTIONS TO IT
\usepackage{caption} % DO NOT CHANGE THIS AND DO NOT ADD ANY OPTIONS TO IT
\usepackage{algorithm}
\usepackage{algorithmic}

\usepackage{newfloat}
\usepackage{listings}
\DeclareCaptionStyle{ruled}{labelfont=normalfont,labelsep=colon,strut=off} % DO NOT CHANGE THIS
\floatstyle{ruled}
\newfloat{listing}{tb}{lst}{}
\floatname{listing}{Listing}
\newcommand{\cnd}{\mskip 0.5mu{|}\mskip 0.5mu}
\DeclareMathOperator{\KS}{\mathrm{C}\mskip 0.1mu}

\let\le=\leqslant
\let\ge=\geqslant

\newtheorem{theorem}{Theorem}
\newtheorem{prop}{Proposition}
\newtheorem{lemma}{Lemma}

\newtheorem{remark}{Remark}

\title{Optimal bounds for dissatisfaction in perpetual voting}
\author{
   Alexander Kozachinskiy~\textsuperscript{\rm 1}, Alexander Shen~\textsuperscript{\rm 2}, Tomasz Steifer~\textsuperscript{\rm 34}
}
\affiliations{
     \textsuperscript{\rm 1}Centro Nacional de Inteligencia Artificial, Santiago, Chile\\
    \textsuperscript{\rm 2}LIRMM, Univ Montpellier, CNRS, Montpellier, France\\
   \textsuperscript{\rm 3}Institute of Fundamental Technological Research, Polish Academy of Sciences, Warsaw, Poland\\%
   \textsuperscript{\rm 4}Pontificia Universidad Cat\'olica de Chile, Santiago, Chile\\%
}

\usepackage{bibentry}
\begin{document}

\maketitle

\begin{abstract}
In perpetual voting, multiple decisions are made at different moments in time. Taking the history of previous decisions into account allows us to satisfy properties such as proportionality over periods of time. In this paper, we consider the following question: is there a perpetual approval voting method that guarantees that no voter is dissatisfied too many times? We identify a sufficient condition on voter behavior ---which we call 'bounded conflicts' condition---under which a sublinear growth of dissatisfaction is possible. We provide a tight upper bound on the growth of dissatisfaction under bounded conflicts, using techniques from Kolmogorov complexity. We also observe that the approval voting with binary choices mimics the machine learning setting of prediction with expert advice. This allows us to present a voting method with sublinear guarantees on dissatisfaction under bounded conflicts, based on the standard techniques from prediction with expert advice. \end{abstract}


\section{Introduction}
Imagine a group of friends who meet every week to go somewhere together. There are several options -- going to the park, going to the cinema, and so on, but different people like only some of the options. For example, somebody says: ``I don't want to go to the park because it is spring and I have an allergy. And I don't want to go to the cinema because I don't like any of the movies that are currently showing''. Then the other friend says: ``I also don't want to go to the cinema, and I don't want to go to dances, because I broke my leg''. And so on, all friends indicate which option they \emph{approve} and which \emph{disapprove}. We have to choose one option for everybody to go there. People, not approving this option, will be \emph{dissatisfied}. %It might be that no option is approved by everybody, then dissatisfied people are inevitable.

This happens not once but a number of times and the preferences of friends might arbitrarily change (for instance, a person who did not want to go to the cinema now might want to see a new movie). When we make a decision, we assume that the preferences of friends in the future are not known to us -- we only see what they approve this week and what they wanted in previous weeks.

%Each time, we choose one option, knowing all the previous history. In particular, we know, how many times each friend was dissatisfied, so we might try to listen more to people who were dissatisfied more often (otherwise, they might complain that it is ``unfair '' for them).

This setting has been recently introduced by \citet{lackner2020perpetual} under the name \emph{perpetual voting}. The goal here is to devise a voting method that would lead to ``fair'' results. For instance, what would be fair in the following situation: 8 friends meet each weekend to go for dinner, but there are just 2 places, and 5 friends want to go to a pizza place and 3 to a curry place (never changing their preferences and approving just one of the options)? It is natural to say that we have to choose the pizza place roughly 5/8 fraction of times, proportionally to the number of people, wanting it.

Using simple majority vote is not a good idea here, it will choose the pizza place all the time. Lackner gives an example of the algorithm that will work better: each time, define the weight of a person as $1/(1 + s)$, where $s$ is the number of times this person was satisfied, and choose the place with a bigger sum of the weights of people that want it. One can show that out of every 8 times, 5 times it will choose the pizza place and 3 times the curry place.

More generally, Lackner introduced \emph{simple proportionality}, which is fairness in the following sense: in a situation when preferences do not change, and everybody approves just one option, every option has to be chosen the number of times which is proportional to the number of people, approving this option. \citet{lackner2023proportional} studied simple proportionality for two classes of voting methods called \emph{loss-based WAMs} and \emph{win-based WAMs} (WAM = weighted approval method). A loss-based WAM is a voting method where every person gets a positive weight, which is a function of the number of times this person was \emph{dissatisfied}, and then the option with the biggest total weight is chosen. A win-base WAM is the same thing, but the weight of a person is determined by the number of times this person was \emph{satisfied}. Lackner and Maly show that there is no simply proportional loss-based WAM, and they characterize all win-based WAMs.

%There are interesting simply proportional methods that are not WAM. For instance, Lackner~\cite{lackner2020perpetual} shows that the following method called Perpetual Consensus is simply proportional: imagine that every person has a bank account, with a potentially negative balance. In every round, we choose an option with the biggest sum of balances of people, approving this option. Now, balances are updated as follows: in every round, every person gets one coin on their account, but satisfied  people have to pay back these deposits in equal shares so that the sum of balances stays the same (if there $n$ people, and $s$ of them are satisfied, dissatisfied people get $+1$ on their account, and satisfied people get $1 - \frac{n}{s}$).

Extending simple proportionality to a setting when agents might approve more than one option, and, moreover, might change their preferences over time, is a delicate task~\cite{bulteau2021justified}, and a recent work of \citet{chandak2024proportional} gives an excellent overview of this topic. Besides proportionality, \citet{lackner2020perpetual}  have formalized and studied other notions of fairness like the independence of uncontroversial decisions (rounds of voting where there is an option, satisfying everyone, should not affect the other rounds)
and dry spells (how many times a person can be dissatisfied in a row).
%Most recently, \citet{neoh2024welfare} looked at perpetual voting methods aimed at maximizing various welfare functions. In particular, they showed NP-completeness of certain decision problem related to maximizing egalitarian welfare and Nash welfare. 

\subsection{Our contribution.} 
In this paper, we introduce a different kind of question, namely:\\
\textbf{Question:} \textit{Is there a voting method that guarantees that each voter is dissatisfied only a small number of times?}

In other words, we want to have a perpetual voting method that allows us to minimize dissatisfaction of every voter. Here, the dissatisfaction is measured as the number of decisions in which the outcome is not approved by the voter. \citet{elkind2024temporal} studied this question in the offline regime, and showed that it is NP-hard to find an optimal way to minimize the dissatisfaction of every voter.
\citet{lackner2020perpetual} considered a related notion of \emph{dry spells} of a voter $v$, that is, sequences of consecutive decisions, during which $v$ is always dissatisfied, and exemplified some methods which guarantee a uniform bound on the length of a dry spell. Our interest is in a somewhat harder task---we want to guarantee that each voter is satisfied often, instead of just asking for each voter to be satisfied at least once in a while.

In the general setting, we face obstacles very quickly.
Imagine a scenario where there are two people, one of whom approves only pizza and the other approves only curry all the time. One of them will be unhappy half the time. If we want a strategy whose guaranteed dissatisfaction for everyone is sublinear in the number of decisions, we have to do something about this.

For that, we introduce a parameter called the \emph{conflict number}. In the case of two alternatives, this parameter can be defined as the maximal number of times a pair of agents does not have a commonly approved option. In the general case, we have to consider the same thing for all subgroups of agents, not exceeding the size of the number of alternatives.
As we just saw, when the conflict number is not bounded by something sublinear in the number of decisions, strategy with the sublinear dissatisfaction is impossible. Surprisingly, we show the converse: if the conflict number is bounded by something sublinear, there is a strategy with sublinear dissatisfaction (ignoring factors that are logarithmic in the number of agents). For that,
we introduce a perpetual voting rule, motivated by a standard machine learning method, the Exponential Weights Algorithm.

We then study the minimal achievable dissatisfaction in a regime when the conflict number is bounded by something negligible compared to the number of decisions.  We derive the optimal bound, which, however, we do not know if one can reach with the Exponential Weights Algorithm or any other computationally efficient strategy. This is because our proof method is non-constructive, relying on inequalities for Kolmogorov complexity.

%The Exponential Weights Algorithms  

%We obtain optimal bounds for this regime, using for the upper bound 

% We also give a nonconstructive existence proof for a voting method with an even better upper bound. Our argument uses Kolmogorov complexity and as such may be of independent interest. This bound is optimal.

 Finally, we discuss a class of simpler algorithms, containing some of the previously studied voting rules (Simple Majority and Perpetual Equality \cite{lackner2020perpetual}), and show that they fail to guarantee sublinear dissatisfaction, even under the bounded conflicts condition.

  \section{Formal setting and contributions} Following~\cite{lackner2020perpetual}, by perpetual voting with $k$ options, $N$ agents, and $T$ rounds, we mean the following game, played between two players that we will call the Decision Maker and the Adversary in rounds $r  = 1, \ldots, T$, where in the $r$-th round:
\begin{itemize}
    \item the Adversary picks $N$ sets $S_1^{(r)}, \ldots, S_N^{(r)}\subseteq \{1, \ldots, k\}$, where $S_i^{(r)}$ is  understood as the set of options, approved by the $i$-th agent in the $r$-th round;
    \item the Decision Maker picks an option $\theta^{(r)}\in\{1, \ldots, k\}$. 
\end{itemize}
We note that this setting does not assume that the set of alternatives is the same  at each round. We only denote by $k$ the maximal number of options in one round (and it can be less than $k$ in other rounds) and arbitrarily index these alternatives by numbers from $1$ to $k$ in each round, while at different rounds these can be essentially different alternatives (like choosing a restaurant to go to one day and a park to go to  
the other day)

For any play of this game, we define the \emph{dissatisfaction} of the $i$-th agent in this play as the number of rounds where the option, chosen by the algorithm, was not approved by this agent:
\begin{align*}
D_i &=\mathbb{I}\{\theta^{(1)}\notin S_i^{(1)}\} + \ldots+\mathbb{I}\{\theta^{(T)}\notin S_i^{(T)}\},\\ i &= 1, \ldots, N.
\end{align*}

In this paper, we initiate the study of strategies for the Decision Maker that aim to guarantee low dissatisfaction for all agents, i.e., to minimize $\max_{i =1}^N D_i$. If the Adversary is unrestricted, it can simply choose all sets to be empty every round, making every agent dissatisfied every time. Therefore, it makes sense to put some restrictions on the Adversary. In what follows, we exactly identify structural restrictions on the game under which a strategy of the Decision Maker with $o(T)$ dissatisfaction exists. Here the number of options is assumed to be a constant $k = O(1)$.

%Imagine a group of $k$ agents that cannot be mutually satisfied in any of the rounds, meaning that no alternative can satisfy them all. This can happen, for example, when the first agent in a group always disapproves the first alternative, the second agent always disapproves the second one, and so on. One agent in the group will be dissatisfied at least $T/k$ times. For a constant $k$, this means that no strategy of the Decision Maker can guarantee $o(T)$ dissatisfaction unless we bound the number of times a group of $k$ agents can be mutually unsatisfiable.

%But is this enough? From any group of mutually unsatisfiable agents, one can select an unsatisfiable subgroup of size at most $k$, by picking any agent disapproving the first option (there must be such an agent as otherwise the first option would have worked for everybody), then any agent disapproving the second option, and so on. However, apriori, one could imagine more complex behavioral patterns that can cause linear dissatisfaction. In this paper, we demonstrate that considering subsets of agents of size at most $k$ is enough.

%In fact, instead of subsets, we will consider all ordered $k$-tuples of agents, and we will bound, for each of these tuples, the number of rounds in which the first agent in the tuple disapproves the first option, the second agent disapproves the second option, and so on.

To this end, we introduce the concept of a conflict.
More specifically,
we say that a subset $A\subseteq \{1, 2, \ldots, N\}$ is in a conflict in the $r$-th round if there exists no $\theta\in\{1,2,\ldots, k\}$ such that $\theta\in S_i^{(r)}$ for every $i\in A$ (no option satisfies every agent in the subset $A$). Given some play in the perpetual voting with $k$ options, we define its conflict number as the maximum, over all $S\subseteq\{1,2,\ldots, N\}$ with $|S|\le k$, of the number of rounds  $S$ is in a conflict.

%we will say that a $k$-tuple of (not necessarily distinct) agents $(i_1, \ldots, i_k)$ is in a \emph{conflict} in a round $r$ if $j\notin S^r_{i_j}$, meaning that the $j$-th agent in a tuple disapproves the $j$-th option. Given some play in the perpetual voting with $k$ options, we define its conflict number as the maximum, over all $k$-tuples of agents, of the number of rounds this tuple was in a conflict. Finally, 
%by the  $C$-conflict perpetual voting we understand a modification of the game where the Adversary cannot make a move that makes the conflict number of a play bigger than $C$. 

%\begin{remark}
  %  One can ask, does it make a difference if we define the conflict number in terms of subsets, namely, as the maximum, over all $k$-sized subsets of agents, of the number of rounds this subset was mutually unsatisfiable. On one hand,  the conflict number for tuples can only be smaller than for subsets because the subset of agents in a conflicting tuple is unsatisfiable. On the other hand, they differ by at most the $k!$-factor because one of the $k!$ orderings of an unsatisfiable $k$-size subset gives a conflicting tuple. For a constant $k$, both numbers are thus equal up to a constant factor, and we choose to work with a smaller one, for tuples, making our results slightly stronger than as if we have used subsets. 
%\end{remark}

We start with 
an observation that as long as $N\ge k$, there is no strategy of the Decision Maker in the $C$-conflict perpetual voting, guaranteeing less than $C/k$ dissatisfaction. Namely, assume that for the first $C$ rounds, the Adversary makes everybody approve everything, except that
for $j = 1, \ldots, k$, the $j$-th agent disapproves the $j$-th option. After $C$ rounds, everybody approves everything without exceptions. In each of the first $C$ rounds, one of the first $k$ agents is dissatisfied, making one of these agents dissatisfied at least $C/k$ times. On the other hand, the $C$-conflict condition is trivially fulfilled as only in the first $C$ rounds somebody disapproves something.

This observation implies that when the conflict bound $C$ is linear in $T$, then for constant $k$  no strategy of the Decision Maker can guarantee $o(T)$ dissatisfaction for everybody. We show that, up to the $poly(\ln N, \ln C)$-factor, the converse is also true -- if $C = o(T)$, then there is a strategy of the Decision Maker, guaranteeing the  $o(T)$ dissatisfaction.
\begin{theorem}
\label{thm_kolmogorov}
    For every $k$ there exists a constant $W > 0$ that for every $N, T, C$ there exists a strategy in the $C$-conflict perpetual voting with  $k$ options, $N$ agents, and $T$ rounds that guarantees dissatisfaction at most $T^{1 - 1/k}\cdot C^{1/k} \cdot (\ln N\cdot \ln C)^{W}$.
\end{theorem}

Indeed, for $k = O(1)$ and $C = o(T)$, ignoring the $poly(\ln N, \ln C)$-factor, this upper bound becomes $T^{1 - \frac{1}{k}} (o(T))^{\frac{1}{k}} = o(T)$. 

\begin{remark}
   One can get rid of the assumption that  $C$ and $T$ are known.  Namely, assume first that $T$ is known but $C$ is not. We start running the algorithm assuming $C = 1$. If the maximal dissatisfaction exceeds the upper bound for $C = 1$, we start again with $C = 2$. We continue in this way, increasing $C$ by a factor of $2$ with each reset. The total maximal dissatisfaction is now an exponential series, with the last term being equal to the whole sum, up to a constant factor. In the algorithm, we can never make our estimate of $C$ twice times bigger than the real one, meaning that the last term is bounded, up to a constant factor, by the same expression. In the same way, one can get rid of the knowledge of $T$.
\end{remark}
When $C$ is negligible compared to $T$, the dissatisfaction bound becomes of order $T^{1 - \frac{1}{k}}$. We show that this dependence on $T$ is tight, even for $C  = 1$, and with the number of agents $N$ growing as $T^{1/k}$ so that $\ln N$ is also negligible.

\begin{prop}
\label{prop_lower}
    For every $k$ and $M$, for $N = k \cdot M$ there exists no strategy that for the $1$-conflict perpetual voting with $k$ options, $N$ agents and $T = M^k$ rounds that guarantees dissatisfaction less than $M^{k-1}/k = T^{(k - 1)/k}/k$. 
\end{prop}
\begin{proof}
    Consider a strategy of the Adversary where it divides $N = k M$  agents into $k$ equal groups of size $M$. In each round, for every $i = 1, \ldots, k$, in the $i$-th group there will be one agent, approving everything except the $i$-th option, and all the other agents of the group approve everything. There will be $T = M^k$ rounds, corresponding to the number of ways to choose one agent per group.

    To be in a conflict, a subset $A$ with at most $k$ agents has to have one agent from every group (if there is nobody from the $i$-th group, the $i$-th option will satisfy everybody in $S$). That is, there are exactly $T^k$ subsets $S$ that can be in a conflict, corresponding to all ways to choose one agent per group. For such $S$ to be in a conflict in the $r$-th round, for every $i = 1, \ldots, k$, the agent of the $i$-th group from $S$ has to disapprove the $i$-th option in the $r$-th round. By construction, for every $S$ there exists only  one $r$ with this property. Thus, the Adversary fulfills the 1-conflict condition.
   % is parameterized by an index of a ``representative'' agent for each groups. Representative of the $i$-th group approves all options except the $i$-th one. All the remaining agents approve all options in this round.

    %Every $k$-tuple of agents can be in conflict only if the first agent in this tuple is from the first group, the second agent is from the second group, and so on. Moreover, this can only happen to this tuple in a single round indexed by this tuple. This means that this strategy of the Adversary satisfies the $1$-conflict restriction.

    We now show that regardless of the choices of the Decision Maker, there will be an agent, dissatisfied at least $M^{k - 1}/k$ times.
    Let $\theta\in\{1, \ldots, k\}$ be the most frequent option, chosen by the Decision Maker. It appears in at least $M^k/k$ rounds. On the other hand, every round involves exactly one agent from the $\theta$-th group, with $M^{k-1}$ rounds for each of $M$ agents of this group (corresponding to $M^{k-1}$ choices of agents from other groups). In one of this $M^{k-1}$-size groups of rounds, in at least $(1/k)$-fraction of rounds the option $\theta$ was elected, because at least this fraction of rounds in total has $\theta$. The agent of $\theta$-th group that appears in this group of rounds will therefore be dissatisfied at least $M^{k-1}/k$ rounds, as required.
\end{proof}
Our proof of Theorem \ref{thm_kolmogorov} uses the Kolmogorov complexity technique which does not yield an efficient strategy achieving this bound. Given $k, N, T, C$, this strategy can be found by a brute-force algorithm, solving the game by analysing all its positions but it requires exponential time and space.

We also give a bound, which has worse dependence on $T$, but is attained by an explicit voting rule,  inspired by the Exponential Weights Algorithm of~\cite{vovk1990aggregating, littlestone1994weighted}. In this rule, every agent gets a weight that is multiplied by a fixed factor each time this agent is dissatisfied, and the rule is to choose an option that minimally increases the sum of the weights.

\begin{theorem}
\label{thm_combi}
    For any $k, N, T, C$, there is a strategy of the Decision Maker, guaranteeing that all agents are dissatisfied at most:
    \[O\left(T^{1 - \frac{1}{k+1}} \cdot (C\cdot k \cdot \ln N)^{\frac{1}{k+1}}\right)\] times in the $C$-conflict perpetual voting with $k$ options, $N$ agents and $T$ rounds.
    \end{theorem}
   Additionally, compared to Theorem \ref{thm_kolmogorov}, this bound does not have an additional polynomial dependence on $\ln C$, and has an explicit small exponent for $\ln N$.

We conclude the paper by analyzing some simpler voting rules and showing that they cannot lead to an $o(T)$ dissatisfaction, even for $k = 2$, $C = 1$, and $N= O(T)$, when $\ln N$ is much smaller than the number of rounds. First, we demonstrate this for the simple majority vote, called \emph{Approval Vote} in \cite{lackner2020perpetual}, where in each round an option with the most approvals is chosen, regardless of the previous history. Second, we show this for the rule called \emph{Perpetual Equality} in \cite{lackner2020perpetual}, which is the majority vote but over agents with maximal dissatisfaction (it can also be seen as the Exponential Weights  Algorithm with a very large factor). In fact, we show this for any \emph{compassionate strategy} of the Decision Maker, which means the following property---if there is a single agent with maximal dissatisfaction approving at least one option, the strategy makes this agent satisfied (i.e., chooses an option approved by them).

\begin{theorem}
\label{thm_lower}
For any $T$, the Approval Vote cannot guarantee dissatisfaction less than $T$  in the 1-conflict perpetual voting with 2 options,  $N = 2T + 1$ agents, and $T$ rounds.

Likewise, for any $T$, no compassionate strategy (including Perpetual Equality) can guarantee dissatisfaction less than $\lfloor T/2\rfloor$ in the 1-conflict perpetual voting with 2 options,  $N = T$ agents, and $T$ rounds.
\end{theorem}

Next three section contain proofs of Theorems \ref{thm_combi},  \ref{thm_kolmogorov}, and \ref{thm_lower}, respectively. 

\section{Proof of Theorem \ref{thm_combi}}
\label{sec_lw}
Our strategy is as follows. Fixing \[\varepsilon  = \left(\frac{\ln N}{T}\right)^{1 - \frac{1}{k+1}} \cdot \left(\frac{1}{Ck}\right)^{\frac{1}{k+1}},\]
the strategy works by assigning a weight to every agent, initially 1 for everybody, that is multiplied by $(1+\varepsilon)$ each time an agent is dissatisfied. The strategy chooses the option that minimally increases the sum of the weights, breaking ties arbitrarily.

We assume that $C\cdot k \cdot \ln N \le T$ because otherwise the stated upper bound on the dissatisfaction is worse than the trivial upper bound of $T$. Hence, $\frac{\ln N}{T}\le \frac{1}{Ck}$, meaning that $\varepsilon \le \frac{1}{Ck}$.

Let $P_r$ be the probability distribution on the agents where the probability of the $i$-th agent is proportional to its weight before the $r$-th round. For instance, $P_1$ is the uniform distribution on agents as all initial weights are the same. Next, let $A_{r,\theta}$ be the set of agents disapproving the option $\theta\in\{1, \ldots, k\}$ in the $r$-th round. Finally, we denote $\delta_{r,\theta} = P_r(A_{r,\theta})$.

If in the $r$-th round the  Exponential Weights Algorithm chooses an option $\theta\in\{1, \ldots, k\}$, then agents from $A_{r,\theta}$ multiply their weights by $(1 + \varepsilon)$. In other words, we add the $\varepsilon$-fraction of the weights of the agents of $A_{r,\theta}$ to the total sum of weights. 
% Dividing the new sum by the total sum of weights before the $r$-th rounds, we see that
This increases the sum of weights by the factor of  
 $(1 + \varepsilon \cdot P_r(A_{r,\theta})) =  (1 + \varepsilon \delta_{r,\theta})$. Hence, the  Exponential Weights Algorithm chooses an option that achieves the minimum:
\[\delta_r = \min\{\delta_{r,1}, \ldots, \delta_{r,k}\},\]
and the total sum of weights gets multiplied by exactly $(1 + \varepsilon \delta_r)$. Therefore, in the end,  the sum of weights will be exactly $N \cdot (1  +  \varepsilon \delta_1)\cdot \ldots \cdot (1  +  \varepsilon \delta_T)$ as the initial sum is $N$. This sum trivially lower bounds the weights of any individual agent, from where we get a bound:
\[(1 + \varepsilon)^{D_i}\le N \cdot (1  +  \varepsilon \delta_1)\cdot \ldots \cdot (1  +  \varepsilon \delta_T),\]
and, after taking the logarithm:
\begin{equation}
\label{eq_D}
D_i \le \frac{\ln N + \sum\limits_{r = 1}^T \ln(1 + \varepsilon\delta_r)}{\ln(1 + \varepsilon)}    
\end{equation}

Since $\varepsilon \le \frac{1}{Ck}\le 1$, we note $\ln(1  +\varepsilon)$ and $\varepsilon$ differ by at most some constant factor, meaning that
\[D_i = O\left(\frac{\ln N}{\varepsilon} + \sum\limits_{r = 1}^T \delta_r.\right)\]

How can we estimate the sum $\sum\delta_r$? We start by bounding the arithmetic mean of $\delta_1, \ldots, \delta_T$ by their $k$-mean:
\[\frac{\sum\limits_{r =1}^T\delta_r}{T} \le \left(\frac{\sum\limits_{r =1}^T\delta_r^k}{T}\right)^{1/k}.\]
We then bound  $\delta_r^k$  by the product $\delta_{r,1}\cdot\ldots \cdot \delta_{r,k}$ as $\delta_r$ by definition is the minimum of the factors in this product, getting:
\begin{equation}
    \label{eq_prod}
    \sum\limits_{r =1}^T\delta_r \le T^{1 - 1/k} \cdot \left(\sum\limits_{r =1}^T\delta_{r,1}\cdot\ldots \cdot \delta_{r,k}\right)^{1/k}.
\end{equation}
    The product $\delta_{r,1}\cdot\ldots \cdot \delta_{r,k}$ is equal by definition  to the product of probabilities
    \[P_r(A_{r,1})\cdot\ldots\cdot P_r(A_{r,k}),\]
    which is also the probability of the Cartesian product $\mathcal{C}_r = A_{r,1}\times \ldots \times A_{r, k}$ w.r.t. the probability distribution on the set of $k$-tuples of agents, obtained by choosing each agent in the tuple independently from $P_r$. We denote this distribution on $k$-tuples by $P_r^{\otimes k}$. This allows us to rewrite \eqref{eq_prod} as 
\begin{equation}
    \label{eq_prod2}
    \sum\limits_{r =1}^T\delta_r \le T^{1 - 1/k} \cdot \left(\sum\limits_{r =1}^T P_r^{\otimes k}(\mathcal{C}_r)\right)^{1/k}.
\end{equation}
Consider any $k$-tuple of agents $(a_1, \ldots, a_k)\in \mathcal{C}_r$. By definition, agents $a_1, \ldots, a_k$ disapprove the 1st, ..., the $k$-th option, respectively, in the $r$-th round. This means that the set $\{a_1, \ldots, a_k\}$ is in the conflict in the $r$-th round.
%The Cartesian product $\mathcal{C}_r$ is exactly the set of $k$-tuples that are in the conflict in the $r$th round. 
Hence, due to the $C$-conflict condition, every $k$-tuple belongs to at most $C$ sets among $\mathcal{C}_1, \ldots, \mathcal{C}_T$. If probability distributions $P_r^{\otimes k}$ were all the same for different $r$, the sum of probabilities $\sum_{r =1}^T P_r^{\otimes k}(\mathcal{C}_r)$ would be bounded by $C$ as every probability of an individual $k$-tuple would appear in this sum at most $C$ times. However, these probability distributions can be different, and to obtain the desired bound, we will use the fact that they change just a little from one round to another.

Namely, the $P_r$-probability and the $P_{r+1}$-probability of any agent differ by at most the $(1+\varepsilon)$-factor. Indeed, at any round, the probability of an agent is computed as its current weights divided by the sum of all weights. In one round, both the numerator and the denominator do not decrease but can increase by at most the $(1+\varepsilon)$-factor. Hence, the fraction can increase by at most the $(1+\varepsilon)$-factor and decrease by at most the same factor.

This means that the $P_r^{\otimes k}$-probability  and the $P_{r+1}^{\otimes k}$-probability of any individual $k$-tuple of agents differ by at most the factor of $\alpha =(1+\varepsilon)^k$. Hence, on an interval of  $l = \lceil\ln_\alpha e\rceil$ 
 rounds, these probabilities can change by the factor at most $\alpha^{l -1}\le e$. We claim that the sum of $P_r(\mathcal{C}_r)$ within any such interval is upper bounded by $e \cdot C$. Indeed, for any $r_0$, the sum $P_{r_0}(\mathcal{C}_{r_0})  + \ldots + P_{r_0 +\ell - 1}(\mathcal{C}_{r_0+\ell - 1})$ is bounded by $e\cdot (P_{r_0}(\mathcal{C}_{r_0})  + \ldots + P_{r_0 }(\mathcal{C}_{r_0+\ell - 1}))$. In turn, the sum $P_{r_0}(\mathcal{C}_{r_0})  + \ldots + P_{r_0 }(\mathcal{C}_{r_0+\ell - 1})$ is bounded by $C$ because any $k$-tuple belongs to at most $C$ sets $\mathcal{C}_{r_0}, \ldots, \mathcal{C}_{r_0+\ell - 1}$.

Splitting all $T$ rounds in $O(T/l)$ intervals of length at most $l$, we get a bound
 \begin{align*}
 \sum\limits_{r =1}^T P_r^{\otimes k}(\mathcal{C}_r) &=  O\left(\frac{CT}{l}\right) =  O\left(\frac{CT}{\ln_\alpha(e)}\right)\\
 &=   O\left(CT\ln(\alpha)\right) = O(k\cdot C \cdot T \cdot \varepsilon).
 \end{align*}
Combining this bound with \eqref{eq_D}
 and \eqref{eq_prod2}, we finally get our upper bound on the dissatisfaction:
 \[D_i = O\left(\frac{\ln N}{\varepsilon} 
 + T (kC\varepsilon)^{1/k}, \right)\]
 which for our choice of $\varepsilon = \ \left(\frac{\ln N}{T}\right)^{1 - \frac{1}{k+1}} \cdot \left(\frac{1}{Ck}\right)^{\frac{1}{k+1}}$ (taken, of course, to make both terms to be equal to each other), transforms into the desired upper bound:
 \[D_i = 
 O\left(T^{1 - \frac{1}{k+1}} \cdot (C\cdot k \cdot \ln N)^{\frac{1}{k+1}}\right).\]

 \section{Proof of Theorem \ref{thm_kolmogorov}}
\label{sec_kolmogorov}
We start by introducing Kolmogorov complexity~\cite{shen2022kolmogorov}. For any partially computable $D\colon\{0, 1\}^*\times\{0, 1\}^*\to\{0,1\}^*$ we define the conditional Kolmogorov complexity $\KS_D(x|y)$ of $x$ given $y$ for two binary strings $x,y\in\{0,1\}^*$ w.r.t.~$D$ as 
\[\KS_D(x\cnd y) = \min\{|p| \colon p\in\{0, 1\}^* \text{ s.t. } D(p, y) = x\},\]
where $|p|$ denotes the length of a binary string $p$. In other words, we look for the shortest description $p$ of $x$ assuming  $y$ is known, where $D$ is used as a \emph{decompressor} that is supposed to produce $x$ from its ``compressed`'' description $p$ assuming the knowledge of $y$.

Different $D$ lead to different ``Kolmogorov complexities'', but there exists an ``optimal'' decompressor $D_{opt}$ for which the resulting complexity function is minimal over all $D$ with the $O(1)$-precision. 
More precisely~\cite[Theorem 17]{shen2022kolmogorov}, there exists $D_{opt}$ such that for any other $D$ there exists a constant $C$ such that:
\[\KS_{D_{opt}}(x\cnd y)\le \KS_{D}(x\cnd y) + C\]
for all strings $x, y\in\{0, 1\}^*$. We fix any optimal decompressor $D_{opt}$ and define the conditional Kolmogorov complexity of $x$ given $y$ as $\KS(x|y) = \KS_{D_{opt}}(x\cnd y)$. 
We also define the unconditional Kolmogorov complexity $\KS(x)$ of a string $x$ as $\KS(x\cnd \Lambda)$, where $\Lambda$ is the empty string.

We will use the observation that for any $y$, the number of strings with $\KS(x\cnd y) < k$ is less than $2^k$. This is because there are less than $2^k$ descriptions of length less than $k$.

We can extend the notion of Kolmogorov complexity from binary strings to any finite objects, like tuples of strings, sets of strings, natural numbers, and so on. It only takes to fix a computable bijection with the set of finite objects of the type we are interested in and the set of binary strings. Then, in place of these finite objects we use their images under this bijection when working with their Kolmogorov complexities. It can be observed that different computable bijections lead to complexity functions that differ only by $O(1)$-term. 

If $x$ and $y$ are two finite objects, we can then define the complexity of their ordered pair $\langle x, y\rangle$. To simplify the notation, we will simply write $\KS(x, y)$ in place of  $\KS(\langle x, y\rangle)$. Likewise, we will use the comma-separated notation for the complexity of tuples of finite objects.

We will use an equality known as the \emph{chain rule}~\cite[Theorem 21]{shen2022kolmogorov} which states that for any two binary strings $x, y$, up to an $O(\log \max\{\KS(x), \KS(y)\})$-term, we have:
\[\KS(x, y) = \KS(x) + \KS(y\cnd x) = \KS(y) + \KS(x\cnd y). \]

 We define the notion of the \emph{mutual information} between two binary strings $x$ and $y$:
\[I(x:y) = \KS(x) + \KS(y) - \KS(x, y).\]
Due to the chain rule, up to an $O(\log \max\{\KS(x), \KS(y)\})$-term, the mutual information can also be written as:
\[I(x:y) = \KS(x) - \KS(x\cnd y) =  \KS(y) - \KS(y\cnd x).\]

For a finite set $A$ we introduce its ``irregularity parameter'' $i(A) = \max\{\KS(A\cnd x) : x\in A)\}$ as the maximal complexity of $A$ given its element.

We need the following inequality. For $k = 2$, it was established in a weaker form by Romashchenko and Zimand~\cite{romashchenko2019operational}.
\begin{prop}
\label{rom_gen}
    For any $k$ there exists $\Gamma > 0$ such that for any $n$, for any  set $\Pi$ of the form $\Pi= 
 U_1\times \ldots \times U_k$ for $U_1, \ldots, U_k \subseteq\{0, 1\}^n$, and for every tuple $(x_1, \ldots, x_k)\in \Pi$, we have:
\begin{align*}
\KS(\Pi\cnd x_1) +& \ldots + \KS(\Pi\cnd x_k)\\ &\le (k - 1) \KS(\Pi) + i(\Pi) + \Gamma\cdot(\log (i(\Pi) + n))). 
\end{align*}
\end{prop}
\begin{proof} We treat $k$ as a fixed constant so that all constants in the $O(\cdot)$-notation might depend on $k$.

We will work with complexities of some binary strings of length $n$ and with complexity of the ``combinatorial parallelepiped''  $\Pi = U_1  \times \ldots \times U_k$. How large can be these complexities? Complexity of any $n$-length binary string $x$ is bounded by $n + O(1)$. Complexity of $\Pi$ can be bounded by $i(\Pi) + O(n)$. This is because one can specify $\Pi$ by any $k$-tuple of $n$-bit binary strings, belonging to it, plus the optimal description of $\Pi$, given this tuple. The first part takes $O(n)$ bits, the second takes at most $i(\Pi)$ by definition of the irregularity parameter. Hence, whenever we use the chain rule, it holds with the $O(\log(i(\Pi) + n))$-precision.

We use the \emph{Romashchenko typization trick}~\cite{romphd} and switch from some object $s$ to a set of objects ``similar to $s$'' (that includes $s$ itself).  For example, a string $s$ of some complexity $m=\KS(s)$ is an element of the set of all strings that have complexity at most $m$. (We say ``at most $m$'' and not ``exactly $m$'' since we will need to enumerate those strings.)

For a pair of strings $(u,v)$ we might consider all pairs $(u',v')$ such that  $\KS(u')\le \KS(u)$, $\KS(v')\le \KS(v)$, $\KS(u',v')\le \KS(u,v)$, $\KS(u'\cnd v')\le \KS(u\cnd v)$,  $\KS(v'\cnd u')\le \KS(v\cnd u)$ (all complexities and conditional complexities of $u',v'$ are bounded by the corresponding complexities of $u$ and $v$; these pairs can be enumerated).

In our case, for every $j = 1, \ldots, k$, we consider objects similar to $x_j$ in the context of a given ``combinatorial parallelepiped'' $\Pi=U_1\times\ldots U_k$. Namely, we consider the set $X_j$ of all $x'_j\in U_j$ such that all the quantities
\begin{equation}\label{eq:2}
 \KS(x_j),  \KS(x_j\cnd\Pi), \KS(x_j,\Pi), \KS(\Pi\cnd x_j)
\end{equation}
do not increase when $x_j$ is replaced by $x'_j$.  Note that we consider only $x'_j\in U_j$ (belonging to the $j$-th projection of the parallelepiped), and use the entire parallelepiped (and not only its $j$-th projection $U_j$) in these expressions.

Obviously, the set $X_j$ contains $x_j$, so it is not empty. On the other hand, its log size is bounded by $\KS(x_j\cnd \Pi)$, since its elements have at most this complexity given $\Pi$. The crucial observation is that \emph{this bound is $O(\log(i(\Pi) +  n))$-tight}:
\begin{equation}
\label{eq_x}
    \log_2 |X_j| = \KS(x_j\cnd\Pi) + i(\Pi) + O(\log(i(\Pi) +  n)).
\end{equation}
Indeed, knowing $\Pi$ and  numerical parameters (complexities in \eqref{eq:2}), we can efficiently enumerate $X_j$ by running the optimal decompressor on all inputs, eventually finding all necessary Kolmogorov complexity upper bounds for all elements of $X_j$.   Thus,  knowing $\Pi$, one can specify $x_j$ itself by its index in this enumeration, which takes $\log_2 |X_j|$ bits, and by numerical parameters, which take $O(\log(i(\Pi) +  n))$ bits, giving us the inequality $ \KS(x_j\cnd\Pi) \le \log_2|X_j| + O(\log(i(\Pi) +  n))$, leading to \eqref{eq_x}.

Next, we claim the following: for any fixed $\delta < 1$, for at least $\delta$-fraction of $x_j'\in X_j$, we have that all complexities in \eqref{eq_x} are the same for $x_j'$ and $x_j$, up to an  $O(\log(i(\Pi) +  n))$-term, with the constant in the $O(\cdot)$-notation depending on $\delta$. It is enough to show this for the conditional complexities   $\KS(x^\prime_j\cnd\Pi)$ and $\KS(x_j\cnd\Pi)$.  Indeed, once we know that, we first can establish the equality (with the same precision) between the complexities of pairs $\KS(x_j, \Pi)$ and $\KS(x_j^\prime, \Pi)$ by writing
\begin{align*}
\KS(x'_j,\Pi)=&\KS(\Pi)+\KS(x'_j\cnd\Pi)\\
&=\KS(\Pi)+\KS(x_j\cnd\Pi)=\KS(x_j,\Pi).
\end{align*}
To establish the approximate equality between $\KS(x_j)$ and $\KS(x_j^\prime)$, and between $\KS(\Pi|x_j)$ and $\KS(x_j\cnd x_j^\prime$, we first notice that they sum up to (approximately) the  same value $\KS(x'_j,\Pi) = \KS(x_j,\Pi)$. On the other hand, in the sum with $x_j^\prime$ both terms do not exceed the corresponding terms in the sum for $x_j$, by definition of $X_j$. This means that the corresponding terms are actually approximately equal.

It remains to show the approximate equality between  $\KS(x^\prime_j\cnd\Pi)$ and $\KS(x_j\cnd\Pi)$.
We have $\KS(x^\prime_j\cnd\Pi)\le \KS(x_j\cnd\Pi)$ by definition for all $x^\prime_j\in X_j$. Now, take the $(1 - \delta)$-fraction of strings of $X_j$ with the lowest complexity given $\Pi$, and let $\ell$ be their maximal complexity so that at least the $\delta$-fraction of strings of $X_j$ have complexity at least $\ell$, leaving us with the task of lower bounding $\ell$.  In this  $\delta$-fraction  of strings there are at most $2^{\ell + 1}$ strings (as all of them have complexity less than $\ell + 1$ given $\Pi$), meaning that in all $X_j$ there are at most $(1/(1-\delta)) \cdot 2^{\ell+1} = O(2^\ell)$ strings. Knowing the bound  $\log_2 |X_j|= \KS(x_j\cnd\Pi) + O(\log(i(\Pi) +  n))$, we conclude that $\ell \ge\KS(x_j\cnd\Pi) + O(\log(i(\Pi) +  n))$.

For the rest of the argument, we  choose $\delta = 1 - \frac{0.01}{k}$. 
Next,  we sample a tuple $(x'_1,\ldots, x'_k)\in X_1\times\ldots\times X_k$ uniformly at random. With positive probability, we have that for $j = 1, \ldots, k$, all the quantities in \eqref{eq:2} are the same both for $x^\prime_j$ and $x_j$ with the $O(\log(i(\Pi) + n))$ precision, and that $\KS((x'_1,\ldots, x'_k)\cnd\Pi) \ge \log_2 |X_1\times\ldots \times X_k| - 10$. Indeed, for any $j = 1, \ldots, k$, the probability that some quantity is too small for $x^\prime_j$ in \eqref{eq:2} is at most $0.01/k$, meaning that probability that there exists a bad $j$ is at most $1\%$,
and the probability that $C(x^\prime_1, \ldots, x^\prime_k|\Pi)$ is too small is no more than $1\%$ just because there are too few tuples of low complexity.

We now fix an arbitrary 
$(x'_1,\ldots, x'_k)\in X_1\times\ldots\times X_k$ 
satisfying these properties. Now it suffices to prove the inequality of Proposition 2 for $x'_1, \ldots, x_k'$ in place of $x_1, \ldots, x_k$, because all the terms in this inequality change by at most $O(\log(i(\Pi) + n))$.

As for any tuple, belonging to $\Pi$, we have:
\begin{equation}\label{eq_irr}
    \KS(\Pi\cnd x_1', \ldots, x_k') \le i(\Pi).
\end{equation}
Let us from now on skip $O(\log(i(\Pi) + n))$-terms as all the inequalities we will use are true with this precision (and all complexities in question are bounded by $O(i(\Pi) + n)$). By our choice of $x_1', \ldots, x_k'$, and by \eqref{eq_x}, we have:
\begin{align*}
    \KS(x_1', \ldots, x_k'\cnd \Pi) &\ge \log_2(|X_1|) + \ldots \log_2(|X_k|)
\\    &= \KS(x_1|\Pi) + \ldots + \KS(x_k|\Pi)\\
&= \KS(x_1'|\Pi) + \ldots + \KS(x_k'|\Pi)
\end{align*}
Having also the inequality $\KS(x_1', \ldots, x_k'\cnd \Pi) \le \KS(x_1'|\Pi) + \ldots + \KS(x_k'|\Pi)$ (the complexity of a tuple is bounded by the sum of complexities of the strings in this tuple, with a precision logarithmic in the complexities of the strings, see Theorem 16 in \cite{shen2022kolmogorov}), we obtain the equality:
\begin{equation}
\label{eq_ind}
 \KS(x_1', \ldots, x_k'\cnd \Pi) =
\KS(x_1'|\Pi) + \ldots + \KS(x_k'|\Pi).
\end{equation}
The inequality that we are aiming to prove, after adding $\KS(\Pi)$ to both sides, looks like that:
\[\KS(\Pi|x_1') + \ldots +\KS(\Pi|x_k') + \KS(\Pi) \le k\cdot \KS(\Pi) + i(\Pi).\]
By \eqref{eq_irr}, it suffices to prove
$%\begin{align*}$
\KS(\Pi|x_1') + \ldots +\KS(\Pi|x_k') + \KS(\Pi)  \le k\cdot \KS(\Pi) + \KS(\Pi|x_1', \ldots x_k')$. %\end{align*}
The last inequality, by definition of the mutual information, is equivalent to:
\[\KS(\Pi) \le I(\Pi:x_1') + \ldots + I(\Pi:x_k') + \KS(\Pi|x_1', \ldots x_k').\]
Re-writing  each mutual information in the other way, we get:
\begin{align*}
    \KS(\Pi) &+ \KS(x_1'\cnd\Pi) + \ldots + \KS(x_k'|\Pi)\\
    &\le \KS(x_1') + \ldots + \KS(x_k') +\KS(\Pi|x_1', \ldots x_k'). 
\end{align*}
By \eqref{eq_ind}, it is equivalent to:
\begin{align*}
    \KS(\Pi) &+  \KS(x_1', \ldots, x_k'\cnd \Pi)\\
    &\le \KS(x_1') + \ldots + \KS(x_k') +\KS(\Pi|x_1', \ldots x_k')
\end{align*}
The left-hand side, by the chain rule, is equal to $\KS(x_1', \ldots, x_k', \Pi)$. Therefore, we have reduced everything to the following inequality:
\[\KS(x_1', \ldots, x_k', \Pi) \le \KS(x_1') + \ldots + \KS(x_k') +\KS(\Pi|x_1', \ldots x_k').\]
It holds because optimal descriptions of $x_1^\prime, \ldots, x_k^\prime$, followed by an optimal description of $\Pi$ given $x_1^\prime, \ldots, x_k^\prime$, with an $O(\log(i(\Pi) + n))$-precision to indicate lengths of these descriptions, can be turned into a description for the whole tuple 
$x_1', \ldots, x_k', \Pi$.
\end{proof}

We now derive Theorem \ref{thm_kolmogorov} from this inequality. For the proof, it will be convenient to extend the notion of a conflict  from subsets of agents to $k$-tuples of agents. Namely, we say that an ordered $k$-tuple $(a_1, \ldots, a_k)\in\{1, \ldots, N\}^k$ is in the conflict in the $r$-th round if for every $i\in\{1, \ldots, k\}$, the agent $a_i$ disapproves the $i$-th option in the $r$-th round. The \emph{tuple conflict number} of a play is the maximum, over all  $(a_1, \ldots, a_k)\in\{1, \ldots, N\}^k$, of the number of rounds the tuple $(a_1, \ldots, a_k)$ was in the conflict. By the \emph{tuple $C$-conflict} perpetual voting we mean a modification of the game where the Adversary has to keep the tuple conflict number of the  play at most $C$.
\begin{lemma}
\label{lemma_numbers}
    The tuple conflict number  of any play is upper bounded by the conflict number of the play.
\end{lemma}
\begin{proof}
   At any round a tuple $(a_1, \ldots, a_k)$ is in the conflict, the set $\{a_1, \ldots, a_k\}$ is also in the conflict.
\end{proof}
Lemma \ref{lemma_numbers} implies that a strategy of decision maker, guarantying maximal dissatifaction at most $D$ in tuple $C$-conflict perpertual voting, also guarantees dissatifaction at most $D$ in the (subset) $C$-conflict perpetual voting. Hence, it is enough to establish  Theorem \ref{thm_kolmogorov} for the tuple $C$-conflict perpetual voting.

We treat $k$ as a fixed constant. Therefore, all constants in the $O(1)$-notation below might depend on $k$, but not on anything else. Next,  we observe that it is enough to show the theorem when $N$, $T$, and $C$ are powers of $2$. Indeed, to show the bound for arbitrary $N, T, C$, we use the strategy for the smallest powers of $2$, exceeding these numbers. This leads to some constant increase in the bound 
on the dissatisfaction that can be compensated by increasing $W$.

Now, for a given $n, t, c$, our goal is to derive an upper bound on $D_{n,t,c}$ which is the minimal $D$ such that there is a strategy of Decision Maker, guaranteeing dissatisfaction at most $D$ in the tuple $C = 2^c$-conflict perpetual voting with $k$ options, $N = 2^n$ agents, and $T = 2^t$ rounds.

We define an auxiliary algorithm $Alg(n, t, c)$ that on input $(n, t, c)$ works as follows. First, it computes $D= D_{n,t,c}$. It is doable because we simply have to solve a finite perfect-information game, completely given by $n, t, c$. Because of the determinacy of such games, there also exists a strategy of the Adversary proving the minimality of $D_{n,t,c}$, meaning that it guarantees that in any play there will be a $D_{n,t,c}$-dissatisfied agent in the perpetual voting with these parameters. The algorithm $Alg(n,t,c)$ finds this strategy of the Adversary.

Then the algorithm converts this strategy into a strategy of Adversary for the game with the same number of rounds $T$, with the same conflict bound $C$,  also guaranteeing  $D_{n,t,c}$-dissatisfaction, but with  $\widehat{N}$ agents, where $\widehat{N}$  is the smallest power of 2 which is at least $N + kT$. We increase the number of players because we want this strategy of the Adversary to be \emph{tuple injective}, by which we mean that it never has the same set of tuples in a conflict in two different rounds. This can be achieved by using additional $kT$ 'dummy' agents. Let us numerate these agents by $a_{r}^i$ for $r = 1, \ldots, T$, $i = 1, \ldots, k$. The agent $a_{r}^i$ disapproves the $i$-th option in the $r$-th round, and apart from that, this agent approves everything every time.
 This does not increase the tuple conflict number of any play. Indeed, take any tuple that includes a new agent $a^i_r$ on the $j$-th position. This tuple can be in the conflict only once, in the $i$-th round, and only if $i = j$.  Likewise, the strategy still guarantees $D_{n,t,c}$-dissatisfaction by for the initial agents. On the other hand, the $r$-th round is the only round in which the tuple $(a^1_r, \ldots, a^k_r)$ is in the conflict, which implies tuple injectivity.

The algorithm takes the maximal $\ell$ such that $2^\ell < D_{n,t,c}$.
Then the algorithm simulates a play against this injective strategy of the Adversary in the game with $\widehat{N}$ agents according to the following counter-strategy (identifying agents with binary strings of length $\widehat{n} = \log_2 \widehat{N}$). In the $r$-th round, for $\theta\in\{1, \ldots, k\}$, let $S^r_\theta\subseteq \{0,1\}^{\widehat{n}}$ be the set of agents disapproving the $\theta$-th option. Define $\Pi^r = S^r_1\times \ldots \times S^r_k$. Note that $\Pi^r$ is exactly the set of $k$-tuples in a conflict in this round. If $\Pi^r$ is empty, meaning that $S_\theta^r$ is empty for some $\theta\in\{1, \ldots, k\}$, 
we choose the option $\theta$, thus making all agents satisfied. Otherwise, we start obtaining better and better upper bounds on the conditional Kolmogorov complexity by running the optimal decompression on all inputs. If for some $\theta\in\{1, \ldots, k\}$ we find out that $\KS(\Pi^r\cnd x_\theta) < l$ for every agent $x_\theta \in S^r_\theta$, we choose the option $\theta$ and the game continues, unless this was already the last round. 
If it never finds such $\theta$, the algorithm goes into an infinite loop without finishing.

Let us start by observing that the algorithm $Alg(n,t,c)$ cannot terminate all $T$ rounds of the game. This is because the strategy of the Decision Maker that it uses guarantees that every agent is dissatisfied at most $2^{l} < D_{n,t,c}$ times. Indeed, each time an agent $x\in\{0, 1\}^n$ was dissatisfied, it is because of some non-empty  $\Pi^r$  with $\KS(\Pi^r\cnd x) < l$. There are at most $2^{l}$ such $\Pi^r$, and each can appear in at most one round due to tuple injectivity.

Hence, there exists $r \in\{ 1, \ldots, T\}$ such that the algorithm never halts when processing $\Pi^r$. This means that for every option $\theta\in\{1, 2,\ldots, k\}$ there is an agent $x_\theta\in S^r_\theta$ with $\KS(\Pi^r\cnd x_\theta)\ge l$  (otherwise we would eventually have found all optimal upper bounds on the conditional Kolmogorov complexity for some option $\theta$). By Proposition \ref{rom_gen}, we obtain that:
\begin{equation}
\label{eq_ineq}
     k\ell \le  (k - 1) \KS(\Pi^r) + i(\Pi^r) + O(\log( i(\Pi^r) + \widehat{n}))
\end{equation}

Now we bound both $C(\Pi^r)$ and $i(\Pi^r)$. We notice that $\Pi^r$ can be identified knowing $r, n, t, c$, by running $Alg(n,t,c)$ and outputting $\Pi^r$. As $r \le T = 2^t$, we need $t + O(\log (ntc))$ bits for that, obtaining the upper bound $\KS(\Pi^r) \le t + O(\log (ntc))$. We now obtain an upper bound on $i(\Pi^r)$. We notice that any given tuple $(x_1, \ldots, x_k)$ can belong to $\Pi^r$ for at most $C = 2^c$ different $r$ because of the $C$-conflict condition. Hence, to describe $\Pi^r$ given $(x_1, \ldots, x_k)$, we need $O(\log(ntc))$ bits to describe numbers $n, t,c$, and also $c$ bits to describe the index of $\Pi^r$ among all these sets of $k$-tuples that contain our tuple, in the same order in which these sets appear during the work of $Alg(n,t,c)$. This gives an upper bound $i(\Pi^r) \le c + O(\log(ntc))$.

Recalling that $\ell$ was chosen as the maximal $\ell\ge 0$ such that $2^{\ell}< D_{n,t,c}$, meaning that $2^{\ell +  1}\ge D_{n,t,c}$, we get from \eqref{eq_ineq} that
$k  \log_2 D_{n,t,c} \le (k - 1) t + c +  O(\log(c\widehat{n}t))$.
which after the exponentiation gives us an upper bound:
\begin{align*}
    D_{n,t,c} &\le (2^t)^{1 - \frac{1}{k}} (2^{c})^{\frac{1}{k}}\cdot (c\widehat{n}t)^{O(1)} \\
    &= T^{1 - \frac{1}{k}}\cdot (\ln C \cdot \ln \widehat{N}\cdot \ln(T))^{O(1)}.
\end{align*}
We have that $\widehat{N} = O(N + T)$, meaning that $\widehat{N}$ can be replaced by $N$ in the bound.
Finally, we notice how to get rid of $T$ in the logarithm. This is because we may assume that $T\le C\cdot N^k$. Indeed, the number of rounds where we cannot satisfy everyone is bounded by $C N^k$ because any such round has at least one $k$-tuple in a conflict. Therefore, any bound on the dissatisfaction we have for $T = C N^k$ are also true for all larger $T$, by satisfying everybody whenever it is possible and playing according to the optimal strategy for $T = C N^k$.

 \section{Proof of Theorem \ref{thm_lower}}

 As in the introduction,  we assume that agents vote to go either to eat pizza or to eat curry. 
For the simple majority vote, we can make the $(2T +1)$st agent dissatisfied $T$ times by making, in the $r$-th round, the $(2T +1)$-st agent approving only pizza, the $2r - 1$-st and the $2r$-th agent approving only curry, and the rest approving both. Each time, curry gets more votes, and the $(2T+1)$-st agent is dissatisfied all the time. To show that the 1-conflict condition is fulfilled, consider any set of agents $A$ of size at most 2. If $A$ is in the conflict, there has to be an agent, dissatisfied with curry, this has to be the $(2T + 1)$st agent. The other agent in $A$ has to be dissatisfied with pizza, and for every agent, there is at most 1 round like that.

 For the lower bound against any compassionate strategy, we will use the following terminology: an agent becomes ``indifferent'' means that from now on, it can only approve both options. We start by making 1 approving only pizza, 2 approving only curry, and 3, 4, ..., $N$ approving both. The Decision Maker chooses one of the options, making either 1 or 2 dissatisfied. Without loss of generality, assume that 1 was satisfied. We make 1 indifferent, ``forgetting'' about this agent. In the next round, we make 2 approving only pizza, and 3, 4, ..., $N$ approving only curry. The agent 2 is currently a single dissatisfied agent, meaning that any compassionate strategy will satisfy 2, choosing pizza and dissatisfying 3, 4, ..., $N$. We now make 2 indifferent, forgetting it, and repeat the same 2 rounds with $3, 4, ..., N$. In more detail, we maintain an invariant that after $2r$ rounds, agents $1, 2, \ldots, 2r$ are indifferent, agents $2r+1, \ldots, N$ have never been in a conflict with each other and their  dissatisfaction is $r$ and is maximal, and that every size-2 set was in a conflict at most once. Repeating the same two rounds with $2r+ 1$ and $2r+2$ in place of $1, 2$, we maintain the invariant from $r$ to $r + 1$, having in the end dissatisfaction  $\lfloor T/2\rfloor$.

\section{Acknowledgments}
Kozachinskiy is funded  by the National Center for Artificial Intelligence CENIA FB210017, Basal ANID. Shen is funded by the FLITTLA ANR-21-CE48-0023 grant. Steifer received generous support from the Millennium Science Initiative Program - Code ICN17002 and the Agencia Nacional de Investigaci\'on y
Desarrollo grant no. 3230203.

\end{document}